\documentclass[twoside,leqno,twocolumn]{article}  
\usepackage{ltexpprt} 
\usepackage{epsfig}

\newcommand{\eps}{\varepsilon}
\def \remark {\noindent {\bf Remark.} \hskip 5pt}
\def \remarks {\noindent {\bf Remarks.} \hskip 5pt}
\newcommand{\ceil}[1]{\left\lceil {#1} \right\rceil}

\begin{document}

\title{\Large A PTAS for TSP with Neighborhoods Among Fat Regions in 
the Plane\thanks{Partially supported by grants from
the National Science Foundation (ACI-0328930, CCF-0431030, CCF-0528209), 
Metron Aviation, and NASA Ames
(NAG2-1620).}}

\author{Joseph S. B. Mitchell\thanks{Stony Brook University, 
Stony Brook, NY 11794-3600}}

\date{March 26, 2014}

\maketitle


\begin{abstract} \small\baselineskip=9pt 
  The Euclidean TSP with neighborhoods (TSPN) problem seeks a shortest
  tour that visits a given collection of $n$ regions ({\em
    neighborhoods}).  We present the first polynomial-time
  approximation scheme for TSPN for a set of regions given by
  arbitrary disjoint fat regions in the plane.  This improves
  substantially upon the known approximation algorithms, and is the
  first PTAS for TSPN on regions of non-comparable sizes.  Our result
  is based on a novel extension of the $m$-guillotine method.  The
  result applies to regions that are ``fat'' in a very weak sense:
  each region $P_i$ 
has area  $\Omega([diam(P_i)]^2)$, 
 but is otherwise arbitrary.  
\end{abstract}

\section{Introduction}

Consider the following variant of the well-studied traveling salesman
problem (TSP): A salesman wants to meet a set of $n$ potential buyers.
Each buyer specifies a connected region in the plane, his {\em
  neighborhood}, within which he is willing to meet the salesman. The
salesman wants to find a tour of shortest length that visits all of
the buyers' neighborhoods and finally returns to its initial departure
point.  This problem, which is known as the {\em TSP with
  neighborhoods} (TSPN), was introduced by Arkin and
Hassin~\cite{ah-aagcs-94} and is a generalization of the classic
Euclidean Traveling Salesman Problem (TSP), in which the regions, or
``neighborhoods,'' are single points, and consequently is
NP-hard~\cite{gj-cigtn-79,p-etspi-77}.

Our main result is a polynomial-time approximation scheme for regions
that are ``fat'' and disjoint in the plane.  Here, we use a very weak
notion of ``fat'' -- a region is {\em fat} if it contains a disk whose
size is within a constant factor of the diameter of the region. We
make no assumption about the sizes of the regions. Previous PTAS
results were known only for the case in which the regions are nearly
of equal size.  The best prior approximation ratio for fat regions as
defined here was $O(\log n)$; for a much more restrictive notion of
fatness, an $O(1)$-approximation was the best prior
result~\cite{bgklos-tspnv-05}.

Our result settles in the affirmative, and solves much more generally,
an open problem that has been circulating in the computational
geometry community for nearly a decade: Is there a PTAS for TSPN on a
set of disjoint disks (or squares)?

\paragraph{Related Work.}
Geometric versions of the TSP have attracted considerable attention in
the last several years, as it was discovered that the TSP on point
sets in any fixed dimension admits a PTAS, by results of
Arora~\cite{a-ptase-98}, Mitchell~\cite{m-gsaps-99,m-gsaps-97}, and
Rao and Smith~\cite{rs-aggsb-98}.  See the
surveys~\cite{a-asnhg-03,m-gspno-00,m-spn-04}.

The TSP with neighborhoods (TSPN) is one of the challenging problems
that has remained largely ``stuck'' in our ability to approximate
optimal solutions.  The best general method remains an $O(\log
n)$-approximation~\cite{efs-atspin-06,gl-faatn-99,mm-aagtn-95}.
If all regions have the same or comparable diameter, but may overlap, then
$O(1)$-approximations are known~\cite{dm-aatsp-03,efs-atspin-06}.
Recently, the most general version of the problem, in which the
regions are allowed to be arbitrary (overlapping) connected
subsets of the plane, has been shown to be
APX-hard~\cite{bgklos-tspnv-05,ss-catsp-03}, as has the case of
(intersecting) line segments of nearly equal
lengths~\cite{efs-atspin-06}, suggesting that it is very unlikely that
a PTAS exists for these versions of the problem.  However, it is open
whether or not a PTAS may exist for the case, e.g., of {\em disjoint}
connected regions in the plane.

Attempts to apply the Arora/Mitchell methods have resulted in only
limited successes.  In particular, Dumitrescu and
Mitchell~\cite{dm-aatsp-03} have shown that if the regions are all
about the same size, have bounded depth, and are fat
(e.g., if the regions are disks in the plane, with bounded ratio of
largest to smallest, with no point lying in more than a constant
number of regions), then the $m$-guillotine method yields a PTAS for
TSPN in the plane that requires time $n^{O(1)}$.
In a related approach, based on Arora~\cite{a-ptase-98}, Feremans and
Grigoriev~\cite{fg-asggp-05} have also given a PTAS, requiring
$n^{O(1/\eps)}$ time, for TSPN for regions that correspond to disjoint
fat polygons of comparable size in the plane; the authors observe that
their algorithm applies also in higher dimensions.  (Actually, the
regions to be visited may be disconnected sets of points that each lie
within one of the disjoint fat polygons of comparable size; a PTAS for
this generalization also follows from \cite{dm-aatsp-03} for the
2-dimensional setting.)

Using a different technique, mapping the problem to an appropriate
``one-of-a-set'' TSP, de Berg et al.~\cite{bgklos-tspnv-05} give an
$O(1)$-approximation for disjoint ``fat'' convex regions; Elbassioni
et al.~\cite{efms-aaegtsp-05} substantially improved the approximation
factor (as a function of the fatness parameter $\alpha$) and
generalized to the discrete case in which the neighborhoods to be visited may
be arbitrary sets of points, with each set lying within a fat region,
{\em not necessarily convex}. 
Most recently, Elbassioni, Fishkin, and Sitters~\cite{efs-atspin-06} 
give an $O(1)$-approximation algorithm for the discrete case
in which the corresponding regions are intersecting, convex, and fat,
of comparable size. These constant-factor approximations
require a much more restrictive (stronger) notion of ``fat'' than we
use in this paper: it is required that {\em any} disk that is not
fully contained in the region, but with its center in the region, must
have a constant fraction of its area inside the region.  This
definition rules out ``skinny tentacles,'' which are allowed in our
definition of fat regions; \cite{bgklos-tspnv-05} also rule out any
form of overlapping among the regions.  Thus, for the class of regions
considered in this paper, no previous approximation bound better than
$O(\log n)$ was known; we give a PTAS.

The original work on the TSPN was by Arkin and
Hassin~\cite{ah-aagcs-94}, who show that when the neighborhoods are
connected and ``well behaved'' (e.g., disks, or having roughly
equal-length and parallel diameter segments), there is an
$O(1)$-approximation algorithm for the TSPN, with running time
$O(n+k\log k)$, where $n$ is the total complexity of the $k$
neighborhoods.  Further, they prove a form of ``combination lemma''
that allows one to consider unions of sets of well-behaved
neighborhoods; the resulting approximation factor is given by the sum
of the approximation factors obtained for each class individually.

For the general case of connected polygonal neighborhoods, Mata and
Mitchell~\cite{mm-aagtn-95} obtained an $O(\log k)$-approximation
algorithm, based on ``guillotine rectangular subdivisions,'' with time
bound $O(n^5)$.  Gudmundsson and Levcopoulos~\cite{gl-faatn-99} have
obtained a faster method, which, for any fixed $\epsilon>0$, is
guaranteed to perform at least one of the following two tasks
(although one does not know in advance which one will be
accomplished): (1) it outputs in time $O(n+k\log k)$ a tour with
length at most $O(\log k)$ times optimal; or (2) it outputs a tour
with length at most $(1+\epsilon)$ times optimal, in time $O(n^3)$ (if
$\epsilon\leq 3$) or $O(n^2\log n)$ (if $\epsilon>3$).  However, no
polynomial-time method guaranteeing a constant factor approximation is
known for general neighborhoods.

The TSPN problem is even harder if the neighborhoods are disconnected.
See the surveys of Mitchell~\cite{m-gspno-00,m-spn-04} for a summary
of results.  A recent result of Dror and
Olin~\cite{do-coedg-04} shows that the TSPN for neighborhoods that
are pairs of points has no PTAS.  

\paragraph{Preliminaries.}\quad
We often speak of the {\em bounding box}, $BB(X)$, of a set $X$, by which we
will always mean the axis-aligned bounding of $X$.
We say that a region $P$ is {\em $\alpha$-fat} (or simply {\em fat}) if
the area, $area(P)$, is at least $\alpha$ times $[diam(P)]^2$,
where $diam(P)$ is the diameter of $P$.
(For $P$ to be fat, it suffices that the ratio of the 
radius of the smallest circumscribing circle to the
radius of the largest inscribed circle is bounded; however, the notion is more general than this.)
Here, we consider $\alpha$ to be a fixed constant.
Note that this definition of fat
applies to convex as well as nonconvex regions.
Note too that fatness implies that the bounding box of a fat region is
``fat'', meaning that it has a bounded aspect ratio (the ratio of its longer side length to its
shorter side length is bounded).  
We critically use our notion of fatness in an area packing argument
(the proof of Lemma~\ref{lem:lower-bnd}); it allows us to give a lower
bound on the area occupied by a set of regions, in terms of the diameters
of the regions.

The input to our algorithm will be a set ${\cal
  R}=\{P_1,P_2,\ldots,P_n\}$ of $n$ disjoint connected fat regions in the plane.
For convenience, we assume that each $P_i$ is a polygonal region,
specified by its vertices.  (More general regions, e.g., splinegons,
are easily handled as well.)  Our algorithms are polynomial in the
total number of vertices used to specify the input.  

A {\em tour}, $T$, is a cycle that visits each region of ${\cal R}$.
The {\em length} of tour $T$, denoted $\mu(T)$, is the Euclidean
length of the curve $T$.  In the {\em TSP with neighborhoods} (TSPN)
problem, our goal is to compute a tour whose length is guaranteed to
be close to the shortest possible length of a tour.  We let $T^*$
denote any optimal tour and let $L^*=\mu(T^*)$ denote its length.  An
algorithm that outputs a tour whose length is guaranteed to be at most
$c\cdot L^*$ is said to be a {\em $c$-approximation algorithm} and to
have an {\em approximation ratio} of $c$.  A family of
$(1+\epsilon)$-approximation algorithms, parameterized by
$\epsilon>0$, is said to be a {\em polynomial-time approximation
  scheme} (PTAS).

\section{Structural Results and A Lower Bound}

Let ${\cal R}=\{P_1,P_2,\ldots,P_n\}$ be a set of $n$ disjoint fat
connected regions (simple polygons) in the plane.
The following lemma follows readily from the triangle inequality.

\begin{lemma}
\label{lem:basic}
  An optimal tour $T^*$ is a simple polygon having at most $n$
  vertices.
\end{lemma}

We let $R_0$ be a minimum-diameter axis-aligned rectangle that
intersects or contains all regions $P_i$.  
Let $D$ be the diameter of~$R_0$.  Note that $R_0$ is easily computed
in polynomial time by standard critical placement arguments; even more
easily computed is a constant-factor approximation of $R_0$, and this
is sufficient for our purposes.

\begin{lemma}
\label{lem:bound}
$2D\leq L^* \leq n D$.
\end{lemma}

\begin{proof}
  The lower bound on $L^*$ follows from the fact (Fact~1 of
  \cite{ah-aagcs-94}) that the shortest tour visiting all four sides
  of the (axis-aligned) bounding box, $W_0=BB(T^*)$, of $T^*$ has length at
  least twice the diameter of $W_0$; since $T^*$ visits all four
  sides of $W_0$, and $R_0$ has diameter at most that of
  $W_0$, this implies that $L^*\geq 2D$.
  The upper bound follows from the fact that each of the at most $n$
  edges of $T^*$ is at most of length $D$, since any two
  regions are at most this distance apart.
\end{proof}

For a fixed $\eps>0$, let ${\cal G}$ denote the regular grid (lattice)
of points $(i\delta,j\delta)$, for integers $i$ and $j$, where
$\delta=\epsilon D/n$.
Let $\Gamma_i$ be the subset of grid points ${\cal G}$ at distance at most
$\delta/\sqrt{2}$ from region~$P_i$.  Note that $\Gamma_i\neq\emptyset$,
and that it may be that $\Gamma_i=\Gamma_j$ for 
$i\neq j$.

\begin{lemma}
\label{lem:grid}
  Any tour $T$ (of length $L$) that visits ${\cal
    R}=\{P_1,P_2,\ldots,P_n\}$ can be modified to be a tour $T_{\cal
    G}$, of length at most $(1+\eps)L$, that visits
  $\{\Gamma_1,\Gamma_2,\ldots,\Gamma_n\}$.  Similarly, any tour
  $T_{\cal G}$, of length $L_{\cal G}$, that visits
  $\{\Gamma_1,\Gamma_2,\ldots,\Gamma_n\}$ can be modified to be a tour
  $T$, of length at most $(1+\eps)L_{\cal G}$, that visits ${\cal
    R}=\{P_1,P_2,\ldots,P_n\}$.
\end{lemma}

\begin{proof}
  Since $T$ visits some point $p_i\in P_i$, for each region $P_i$,
  we can simply add to $T$ a detour that goes from $p_i$ to a grid
  point and back to $p_i$, for each $i$.  Since no point $p_i\in P_i$
  is further from a grid point of $\Gamma_i$ than $\delta/\sqrt{2}$,
  we get that the total detour length is bounded above by $n\cdot
  2\delta/\sqrt{2} = \eps D\sqrt{2}\leq \eps L^*\leq \eps L$.  The second claim
  is proved similarly.
\end{proof}

A consequence of the lemma is that we can assume, without loss of
generality, that the input regions are each replaced by
grid-conforming rectilinear polygons, with vertices on the grid.


The next lemma provides a means of ``localizing'' an optimal
solution, so that our search for approximately optimal tours can be
restricted to a polynomial-size grid.  
Let $W_0=BB(T^*)$ denote the axis-aligned bounding box of $T^*$, an
optimal tour/tree.
We can assume that $W_0$ contains at least one vertex, $c_0$, of some polygon $P_i$;
otherwise, the problem can be directly solved to optimality in polynomial time.
\footnote{
If $W_0$ contains no vertices of the polygons ${\cal R}$, then we
know that $T^*$ is a shortest tour/tree visiting at least one edge of
each polygon $P_i$ (at points interior to the edges); thus, each $P_i$
has at least one edge intersecting $W_0$.  Since we assume the
polygons $P_i$ are disjoint, we use a natural dominance relationship
among the edges to reduce the problem to the trivial one of finding a
shortest tour visiting at most 4 disjoint line segments (corresponding
to at most one edge per corner of $W_0$).  We solve this case
directly, computing $T^*$ optimally, enumerating the combinatorially
distinct rectangles that have no vertices within them and have at
least one edge of each $P_i$ crossing them.
}
Our algorithm enumerates over all choices of~$c_0$.  

\begin{lemma}
\label{lem:localization}
  There exists an optimal tour $T^*$ of the regions ${\cal R}$, of
  length $L^*$, that lies within the ball, $B(c_0,D_0)$, of radius $D_0=O(nD)$
  centered at $c_0$, a vertex within $W_0$.  Further, there exists a tour $T^*_{\cal G}$ of
  the grid sets $\Gamma_i$, of length at most $(1+\eps)L^*$, that has
  its vertices at grid points ${\cal G}$ that lie within an $N$-by-$N$
  array of grid points centered at $c_0$, where $N=O(n^2/\eps)$.
\end{lemma}

\begin{proof}
Let $R_{c_0}$ be a minimum-diameter rectangle centered at $c_0$ that
intersects (or contains) every region of ${\cal R}$.  Since $R_{c_0}$
has diameter, $D_{c_0}$, no greater than twice the diameter of $W_0$,
we know, from Lemma~\ref{lem:bound}, that $D_{c_0}=O(nD)$.  If $T^*$ has no point within $R_0$,
then there can be no region of ${\cal R}$ interior to $R_{c_0}$; thus,
the boundary, $\partial R_{c_0}$, meets all regions of ${\cal R}$, so
we know that $L^*\leq |\partial R_{c_0}| = O(D_{c_0})=O(nD)$.  This
implies that all of $W_0$ (and thus all of $T^*$) lies within distance
$O(nD)$ of $c_0$.

If $T^*$ enters $R_{c_0}$, then, if the tour were to wander
substantially outside of $R_{c_0}$, say to a point $q$ at distance at
least $\sqrt{2}D_{c_0}$ from $R_{c_0}$, then that portion of the tour
(of length at least $2\sqrt{2}D_{c_0}$) connecting $q$ to the boundary
of $R_{c_0}$ can be replaced with a path along the boundary of
$R_{c_0}$ (whose perimeter is at most $2\sqrt{2}D_{c_0}$, since its
diameter is $D_{c_0}$), while meeting the same set of regions (by
connectivity of the regions).  Thus, $T^*$ does not venture more than
distance $O(D_{c_0})=O(nD)$ from the center, $c_0$, of $R_{c_0}$.

Since grid points of ${\cal G}$ are at spacing $\delta=\eps D/n$, we
see that a grid of size $N=O(nD/\delta)=O(n^2/\eps)$ suffices, and
Lemma~\ref{lem:grid} shows that $T^*$ can be rounded to ${\cal G}$.
\end{proof}

Let $W_0$ be the axis-aligned bounding box of an optimal tour $T^*$.
Let ${\cal R}_{W_0}$ be the subset of regions ${\cal R}$ that lie
entirely inside $W_0$. (Note that each of these regions has diameter
$O(nD)$, since $W_0$ has diameter at most $L^*\leq nD$, by Lemma~\ref{lem:bound}.)  
We partition this set ${\cal R}_{W_0}$ of regions into
$K=O(\log (nD/\delta))=O(\log(n/\eps))$ classes, according to the diameters being in
the intervals $(0,\delta), (\delta,2\delta), (2\delta,4\delta),
(4\delta,8\delta),\ldots(2^{K-2}\delta,2^{K-1}\delta)$.  

Note that, by the argument of Lemma~\ref{lem:grid}, the ``small''
regions in the size class $(0,\delta)$ can effectively be replaced
each by a single grid point, which we insist on the tour visiting;
from now on, we assume that this replacement of small regions has been
done.  Thus, we focus on regions in the size classes
$(2^{i-1}\delta,2^i\delta)$, for $i=1,2,\ldots,K-1$, whose largest
diameter is denoted $d_i=2^i\delta$.  (If there are no such non-small 
regions of ${\cal R}_{W_0}$, then our TSPN instance is easy to solve
approximately, as an instance of TSP on a point set, with the added constraint
of visiting the large regions $\bar{\cal R}$.)
Note too that fatness implies that no point lies in more than a constant number
of bounding boxes of the input regions of any one size class; thus, no point lies
in more than $O(K)=O(\log(n/\eps))$ bounding boxes of regions.

Let $T^* \oplus B(d_i)$ be the Minkowski sum of the ball $B(d_i)$ of
radius $d_i$ centered at the origin and the optimal tour $T^*$.  The
region $T^* \oplus B(d_i)$ is that swept by a ball of radius $d_i$
whose center follows the optimal tour~$T^*$.

\begin{lemma}
\label{lem:area-swept}  
  The area, $A_i$, of $(T^* \oplus B(d_i))\cap W_0$ is at most $2d_i
  L^*$.
\end{lemma}

\begin{proof}
  The area of $T^* \oplus B(d_i)$ is at most $2d_i L^* + \pi d_i^2$.
  Since $W_0$ is a (tight-fitting) bounding box of $T^*$, there is
  some point $p_L\in T^*\cap \partial W_0$ on the left wall of $W_0$
  and some point $p_R\in T^*\cap \partial W_0$ on the right wall of
  $W_0$.  The left half-disk of the radius-$d_i$ ball centered at
  $p_L$ and the right half-disk of the radius-$d_i$ ball centered at
  $p_R$ both lie outside of $W_0$; thus, the area $A_i$ does not
  include (at least) area $\pi d_i^2$ of the region $T^* \oplus
  B(d_i)$.  Thus, $A_i\leq 2d_i L^* + \pi d_i^2-\pi d_i^2=2d_i L^*$.
\end{proof}

From Lemma~\ref{lem:area-swept}, we know that $L^* \geq A_i/2d_i$, for
each $i=1,\ldots,K-1$.  We also know that the sum of the areas of the
$n_i$ regions of class $i$ is at most $A_i$, since the (disjoint)
regions all must lie fully within the Minkowski sum $T^* \oplus
B(d_i)$.  Since each region of class $i$ has diameter at least $d_i/2$, 
and, by fatness, has area $\Omega(d_i^2)$,
we get that $A_i\geq C_0 d_i^2 n_i$, for an appropriate constant
$C_0$.  Thus, $L^*\geq (C_0/2)d_i n_i$, for each $i$.  Summing
on $i$, we get $K L^*\geq (C_0/2)\sum_i d_i n_i$.  This implies that
$L^*\geq C\cdot \lambda({\cal R}_{W_0})/\log (n/\eps)$, for
some constant $C$, where $\lambda({\cal R}_{W_0})$ is the sum of the
diameters of the regions ${\cal R}_{W_0}$.

\begin{lemma}
\label{lem:lower-bnd}
  $L^*\geq C\cdot \lambda({\cal R}_{W_0})/\log (n/\eps)$, for some
  constant~$C$.
\end{lemma}

\remarks 

(1). The bound of Lemma~\ref{lem:lower-bnd} can be improved to
  $L^*\geq C\cdot \lambda({\cal R}_{W_0})/(\log n)$, by the following
observation (thanks to Khaled Elbassioni and Rene Sitters):
It suffices for the bound to consider only disks of
diameter greater than $D/n$, since the sum of all diameters of regions
with diameter less than $D/n$ is at most $D\leq L^*/2$.  Thus, it suffices to
consider regions within the range of diameters $[D/n,D]$, for which there are only $O(\log n)$ intervals
$(2^{i-1}\delta,2^i\delta)$.

(2). The bound of Lemma~\ref{lem:lower-bnd} is asymptotically
tight, as can be seen in Figure~\ref{fig:tight}.

\begin{figure}[htbp]
 \centering{\includegraphics[height=1.5in]{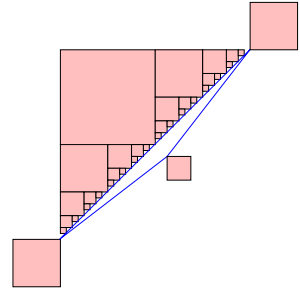}}
\caption{An optimal tour on a set of disjoint squares
for which $L^*= C\cdot \lambda({\cal R}_{W_0})/\log n$.}
 \label{fig:tight}
 \end{figure}

\section{Approximation Scheme}

It is natural to suspect that the same techniques that yield a PTAS
for geometric TSP on points may apply to the TSPN.  The basic issue we
must address in order to apply these techniques is to be able to write
a recursion to solve an appropriate ``succinct'' subproblem with
dynamic programming.  

What should a subproblem be ``responsible'' to solve?  For the TSP on
{\em point} data, the subproblem is responsible for constructing an
inexpensive network (of a particular special structure) on the points
that are {\em inside} the subproblem rectangle, and to interconnect
this network with the boundary in some nicely controlled way (e.g.,
with only a constant complexity of connection, in the case of
$m$-guillotine methods).  The problem with the TSPN is that the {\em
regions} can cross subproblem boundaries, making it difficult to
specify which of the regions is subproblem responsible to visit versus
which of the regions are visited {\em outside} the subproblem.  We
cannot afford to enumerate subproblems corresponding to all possible
subsets of regions that cross the boundary of the rectangle.

Our new idea is to introduce an extension of the general method of
$m$-guillotine subdivisions to $m$-guillotine subdivisions that
include not only a ``bridge'' for the ``$m$-span'' of each cut, but
also a ``region-bridge'' for the ``$M$-region-span'' of the set of fat
regions, with $M=O((1/\eps)\log (n/\eps))$.  The standard analysis of the
$m$-guillotine method allows us to charge off the construction cost of
the $m$-spans, while the new analysis we gave in
Lemmas~\ref{lem:basic}-\ref{lem:lower-bnd} allows us to charge off the
additional cost associated with the $M$-region-span to the sum of the
sizes (diameters) of the regions.  A similar idea, using a
``disk-span'' was employed in the PTAS of \cite{dm-aatsp-03}; however,
the novelty of our analysis is that we are able to avoid the
requirement of equal-size regions, as in \cite{dm-aatsp-03}, by
partitioning the regions into two types (those fully inside $W_0$ and
those that lie partially outside $W_0$), and by analyzing separately
the two types of spans, allowing the $M$ to be logarithmic in~$n$.

For each of the $M$ regions that crosses a cut, we can afford (since
$2^{O(M)}$ is polynomial in $n$) to specify, as part of the corresponding
subproblem, which regions are to be visited inside the subproblem.  It
is key that $M$ is only logarithmic in $n$; our lower bound on the
tour length (Lemma~\ref{lem:lower-bnd}) is ``just right'', in that it
gives us this logarithmic factor.

We review some definitions, largely following the notation of
\cite{m-gsaps-99}.  Let $G$ be an embedding of a planar graph, and let
$L$ denote the total Euclidean length of its edges, $E$. We can assume
(without loss of generality) that $G$ is restricted to the unit
square, $B$ (i.e., $E\subset int(B)$), and that the vertices of $G$
lie at grid points ${\cal G}$.

Consider an axis-aligned rectangle $W$ (a {\em window}) with
$W\subseteq B$ and with corners at grid points.  ($W$ will correspond
to a subproblem in a dynamic programming algorithm.)  Let $\ell$ be an
axis-parallel line, through grid points, intersecting $W$. We refer to
$\ell$ as a {\em cut} for $W$.
We will refer to a {\em root window}, $W_0$, which is a window that is
hypothesized to be the minimal enclosing bounding box of an optimal
grid solution, $T^*_{\cal G}$; as such, $W_0$ necessarily intersects
or contains every region of ${\cal R}$.  All windows $W$ of interest
will then be subwindows of $W_0$.  By Lemma~\ref{lem:localization}, we
know that there are only a polynomial number ($O((n/\eps)^4)$) of
possible choices for $W_0$; we can afford to try each one.  With
respect to a fixed choice of $W_0$, we let ${\cal R}_{W_0}$ denote
those {\em internal regions} that are entirely contained within $W_0$,
and we let $\bar{\cal R}_{W_0}={\cal R}\setminus {\cal R}_{W_0}$
denote the {\em border regions} that are not contained in $W_0$ (but
do meet the boundary of $W_0$).

The intersection, $\ell\cap (E\cap int(W))$, of a cut $\ell$ with
$E\cap int(W)$ (the restriction of $E$ to the window $W$) consists of
a (possibly empty) set of subsegments (possibly singleton
points) of $\ell$.  Let $\xi$ be the number of endpoints of such
subsegments along $\ell$, and let the points be denoted by
$p_1,\ldots,p_{\xi}$, in order along~$\ell$.  For a positive integer
$m$, we define the {\em $m$-span}, $\sigma_m(\ell)$, of $\ell$ (with
respect to $W$) as follows.  If $\xi\leq 2(m-1)$, then
$\sigma_m(\ell)=\emptyset$; otherwise, $\sigma_m(\ell)$ is defined to
be the (possibly zero-length) line segment, $p_{m}p_{\xi-m+1}$,
joining the $m$th endpoint, $p_m$, with the $m$th-from-the-last
endpoints, $p_{\xi-m+1}$.
Line $\ell$ is an {\em $m$-good cut with respect to $W$ and $E$} if
$\sigma_m(\ell)\subseteq E$.  (In particular, if $\xi\leq 2(m-1)$, 
then $\ell$ is trivially an $m$-good cut.)

The intersection of the cutting segment $ab=\ell\cap W$ with the 
bounding boxes of the 
internal regions ${\cal R}_{W_0}$ restricted to $W$
consists of a (possibly empty) set of subsegments. 
Let $\xi$ denote the number of bounding boxes of regions ${\cal R}_{W_0}$ 
that segment $ab$ crosses.  (Note that the endpoints, $a$ and $b$, can each lie
inside $O(K)=O(\log (n/\eps))$ bounding boxes; these ``corner boxes'' are not
counted among the $\xi$ boxes that $ab$ crosses.)
We define the {\em $M$-region-span}, $\Sigma_M(\ell)$, of $\ell$
analogously to the $m$-span: If $\xi<2M-1$, then $\Sigma_M(\ell)$ is
defined to be empty; otherwise, if $\xi\geq 2M-1$, then
$\Sigma_M(\ell)$ is the line segment $a_Mb_M$, along $\ell$, with
$a_M$ defined to be the $M$th entry point where segment $ab$ enters a
bounding box, when going from $a$ towards $b$ along $ab$, and $b_M$
defined similarly to be the $M$th entry point where segment $ab$
enters a bounding box, when going from $b$ towards $a$ along $ab$.
Line $\ell$ is an {\em $M$-good cut with
respect to $W$, $E$, and ${\cal R}_{W_0}$} if $\Sigma_M(\ell)\subseteq E$.

We now say that $E$ satisfies the {\em $(m,M)$-guillotine property
  with respect to window $W$ and regions ${\cal R}_{W_0}$} if either
(1) no edge of $E$ lies (completely) interior to $W$; or (2) there
exists a cut $\ell$, that is $m$-good with respect to $W$ and $E$ 
and $M$-good with respect to $W$, $E$, and ${\cal
  R}_{W_0}$, that splits $W$ into $W_1$ and $W_2$, and, recursively,
$E$ satisfies the $(m,M)$-guillotine property with respect to both
$W_1$ and~$W_2$, and regions ${\cal R}_{W_0}$.

We say that a point $p\in W$ is {\em $m$-dark with respect to horizontal cuts of $W$}
if the vertical rays going upwards/downwards from $p$ each cross at least $m$
edges of $E$ before reaching the boundary of $W$.
Similarly, we say that a point $p\in W$ is {\em $M$-region-dark with respect to horizontal cuts of $W$}
if the vertical rays going upwards/downwards from $p$ each cross at least $M$
bounding boxes of regions of ${\cal R}_{W_0}$ before reaching the
boundary of~$W$.  As in \cite{m-gsaps-99}, the length of the $m$-dark
portion of a cut is the ``chargeable'' length of the cut that is
chargeable to the lengths of the $m$ layers of $E$ on each side of the
cut that become ``exposed'' after the cut.
Similarly, the length of the $M$-region-dark portion of a cut is the chargeable length
of the cut that is chargeable to the $M$ layers of bounding boxes on each side of the cut that
become exposed after the cut. 

Given an edge set $E$ of a connected planar graph $G$, if $E$ is not
already satisfying the $(m,M)$-guillotine property with respect to
$W_0$ and regions ${\cal R}_{W_0}$, then we argue, by the standard
guillotine argument of \cite{m-gsaps-99}, that there exists a
``favorable cut'' for which we can afford to charge off (to the
edges of $E$ and the edges of the bounding boxes of regions) the
construction of any $m$-span or $M$-region-span that must be added to
$E$ in order to make the cut both $m$-good with respect to $W$ and $E$ and
$M$-good with respect to $W$, $E$, and~${\cal R}_{W_0}$:

\begin{lemma}
\label{lem:key}
For any $G$ and any window $W$, there is a favorable cut.
\end{lemma}

\begin{proof}
We show that there must be a favorable cut that
is either horizontal or vertical.

Let $f(x)$ denote the ``cost'' of
the vertical line, $\ell_x$, through $x$,
where ``cost'' means the sum of the lengths of the
$m$-span and the $M$-region-span for $\ell_x$; thus,
$f(x)=|\sigma_m(\ell_x)| + |\Sigma_M(\ell_x)|$.

Then, $A_x=\int_0^1 f(x) dx$ is simply the area, $A^{(m)}_x=\int_0^1
|\sigma_m(\ell_x)| dx$, of the ($x$-monotone) region $R^{(m)}_x$ of
points of $B$ that are $m$-dark with respect to horizontal cuts, plus
the area, $A^{(M)}_x=\int_0^1 |\Sigma_M(\ell_x)| dx$, of
the ($x$-monotone) region $R^{(M)}_x$ of points of $B$ that
are $M$-region-dark with respect to horizontal cuts.  Similarly, define
$g(y)$ to be the cost of the horizontal line through $y$, and let
$A_y=\int_0^1 g(y) dy$.

Assume, without loss of generality, that $A_x\geq A_y$.
We claim that there exists a horizontal favorable cut; i.e., we claim
that there exists a horizontal cut, $\ell$, such that its chargeable
length (i.e., length of its $m$-dark portion plus its $M$-region-dark
portion) is at least as large as the cost of $\ell$
($|\sigma_m(\ell)|+|\sigma_M(\ell)|$).  To see this, note
that $A_x$ can be computed by switching the order of integration,
``slicing'' the regions $R^{(m)}_x$ and $R^{(M)}_x$
horizontally, rather than vertically; i.e., $A_x=\int_0^1 h(y)
dy=\int_0^1 h_m(y) dy+\int_0^1 h_{M}(y) dy$, where $h_m(y)$
is the $m$-dark length of the horizontal line through $y$, $h_{M}(y)$ 
is the length of the intersection of $R^{(M)}_x$
with a horizontal line through $y$, and $h(y)$ is the chargeable
length of the horizontal line through $y$.  (In other words, $h_m(y)$
(resp., $h_{M}(y)$) is the length of the $m$-dark (resp.,
$M$-region-dark) portion of the horizontal line through $y$.)  Thus,
since $A_x\geq A_y$, we get that $\int_0^1 h(y) dy \geq \int_0^1 g(y)
dy\geq 0$.  Thus, it cannot be that for all values of $y\in[0,1]$,
$h(y)<g(y)$, so there exists a $y=y^*$ for which $h(y^*)\geq g(y^*)$.
The horizontal line through this $y^*$ is a cut satisfying the claim
of the lemma. (If, instead, we had $A_x\leq A_y$, then we would get a
{\em vertical} cut satisfying the claim.)
\end{proof}

The charging scheme assigns a charge to the edges of $E$ of total
amount equal to (roughly) $1/m$th of the length of $E$, and it assigns a charge
to the edges of the bounding boxes of regions ${\cal R}_{W_0}$ of
total amount equal to (roughly) $1/M$th of the diameters of the regions.  (Note
that the diameter/perimeter of each bounding box is proportional to
the diameter of the corresponding region.)  We therefore have shown
the following structure theorem:

\begin{theorem}
\label{thm:main-guil}
  Let $G$ be an embedded connected planar graph, with edge set $E$
  consisting of line segments of total length $L$.  Let ${\cal R}$ be
  a set of disjoint fat regions and assume that $E\cap P_i\neq
  \emptyset$ for every $P_i\in {\cal R}$.  Let $W_0$ be the
  axis-aligned bounding box of $E$.  Then, for any positive integers
  $m$ and $M$, there exists an edge set $E'\supseteq E$ that obeys the
  $(m,M)$-guillotine property with respect to window $W_0$ and regions
  ${\cal R}_{W_0}$ and for which the length of $E'$ is at most
  $L+{\sqrt{2}\over m}L+{\sqrt{2}\over M}\lambda({\cal R}_{W_0})$,
  where $\lambda({\cal R}_{W_0})$ is the sum of the diameters
  of the regions ${\cal R}_{W_0}$.
\end{theorem}

In the next section we will give a dynamic programming algorithm to
compute a minimum-length $(m,M)$-guillotine edge set that obeys
certain constraints.  The algorithm works on a discrete
polynomial-size grid, ${\cal G}$, corresponding to the regular grid of
resolution $\delta$ within a (grid-rounded) axis-aligned box $W_0$
that is hypothesized to be the bounding box of an optimal tour.  For a
given edge set $E$, whose edges have endpoints on the grid (as we can
assume is the case for an approximately optimal tour, by
Lemma~\ref{lem:grid}), the proof of the above theorem can be applied
to the {\em grid encasement} of each edge $e\in E$: the {\em
encasement} of $e$ is defined to be the (rectilinear) simple polygon
$Q_e$ consisting of the union of grid cells whose interiors intersect
$e$.  Note that $Q_e$ lies within $BB(e)$ (since the endpoints of $e$
lie on the grid) and that the perimeter of $Q_e$ is at most
$2\sqrt{2}\cdot |e|$, where $|e|$ denotes the Euclidean length of $e$.
(If a set of encasements is spanned, then the corresponding set of
edges is also spanned, since each edge is contained within its
encasement, implying that the span has been rounded outwards.)  Also,
since the regions $P_i$ can be replaced by the grid sets $\Gamma_i$
(Lemma~\ref{lem:grid}), the bounding boxes $BB(\Gamma_i)$ lie on the
grid, and the $M$-region-spans also lie on the grid.  Then, the proof
of Theorem~\ref{thm:main-guil} applies to $\Gamma_i$ and the edge set $\bar
E$ consisting of the (horizontal/vertical) edges bounding all
encasements $Q_e$, for which the functions $f$ and $g$ are
piecewise-constant on the grid, implying that the $m$-spans that are
added to $E$ are always vertical/horizontal segments with endpoints on
the grid.

We note that the edge set $E'$ guaranteed in the above argument need
not be connected (as is the case for $E$), since the region-spans that
we add (and charge off in the charging scheme) may not intersect
intersect edges of $E$.  This issue is readily addressed, as we
describe in the next section (see also \cite{dm-aatsp-03}).

Further, the proof of the above theorem shows also that we can afford
to double (or replicate any constant number of times) the $m$-spans
and $M$-spans that are added to $E$ to obtain an augmented edge set
with the $(m,M)$-guillotine property.  We exploit this fact in making
the usual ``bridge doubling'' argument (see \cite{m-gsaps-99}) that allows
the network we compute with the dynamic program of the next section to
contain an Eulerian subgraph, from which a tour is extracted.

The main result of this paper is summarized in the following theorem:

\begin{theorem}
\label{cor:ptas}    
The TSPN for a set of disjoint fat regions has a PTAS.  
\end{theorem}

\begin{proof} 
  Consider an optimal tour, $T^*$, of length $L^*$.  By
  Lemma~\ref{lem:grid} and Lemma~\ref{lem:localization}, we know that
  there is a grid-rounded tour $T^*_{\cal G}$ of comparable length
  whose vertices lie on a certain polynomial-size grid
  (of size $O(n^2/\eps)$-by-$O(n^2/\eps)$).  We will consider separately each
  choice of $W_0$, the hypothesized bounding box of $T^*_{\cal G}$.
  
  For a given choice of $W_0$, we apply the dynamic programming
  algorithm of the next section to compute a minimum-weight edge set
  $E^*$ that has several specified properties: (a) it is
  $(m,M)$-guillotine with respect to window $W_0$ and regions ${\cal
    R}_{W_0}$, with doubled bridge segments; (b) it 
satisfies certain connectivity requirements (made precise in the next section); 
and, (c) it visits all of the regions ${\cal R}$.
As described in the next section, the network that is output by the
dynamic program can be readily made to be connected and (using the
bridge-doubling) to contain an Eulerian subgraph spanning the regions.
 
  Assuming that $W_0$ is the correct choice of bounding box, by
  Theorem~\ref{thm:main-guil}, we know that the edge set $E$
  corresponding to $T^*_{\cal G}$ has an associated edge set
  $E'\supseteq E$ that satisfies properties (a)-(c) and has length at
  most $L+ O({1 \over m}L)+O({1 \over M}\lambda({\cal
    R}_{W_0}))$, where $L\leq (1+\eps)L^*$ is the length of $T^*_{\cal
    G}$.  Since $E^*$ is a minimum-length edge set satisfying
  conditions (a)-(c), we get then that $E^*$ has length at most
  $L+O({1 \over m}L)+O({1 \over M}\lambda({\cal R}_{W_0}))$.  By
  Lemma~\ref{lem:lower-bnd}, we know that $\lambda({\cal R}_{W_0}) \leq
  (L^*/C)\log (n/\eps)$, for some constant $C$.  Picking
  $m=\ceil{1/\eps}$ and $M=\ceil{(1/\eps)\log (n/\eps)}$, and putting the
  pieces together, we get that $E^*$ has length at most
  $(1+C_1\eps)L^*$, for a constant $C_1$.  
The running time of the algorithm is $2^{O(M)}n^{O(1/\eps)}$, which is
polynomial in $n$, for any fixed~$\eps$, since $M=\ceil{(1/\eps)\log
(n/\eps)}$.
\end{proof}

\section{The Algorithm}

We now describe the dynamic programming algorithm, running in
$2^{O(M)}n^{O(m)}$ time, to compute a minimum-length planar graph
having a prescribed set of properties: (1) it satisfies the
$(m,M)$-guillotine property (necessary for the dynamic program to have
the claimed efficiency); (2) it visits each of the grid point sets
$\Gamma_i$ corresponding to region $P_i$; and (3) it consists of
a connected component and a set of region-bridges, which
can then be augmented 
to be connected and to contain an
Eulerian subgraph that spans the $\Gamma_i$'s (this condition allows
us to extract a tour in the end).  We only outline here the dynamic
programming algorithm, highlighting the modifications to account for
the $M$-region-span; the details are similar to those
of~\cite{m-gsaps-99}.

Our algorithm computes $D$, the diameter of $R_0$, a minimum-diameter axis-aligned rectangle that
intersects or contains all regions $P_i$.  
By Lemma~\ref{lem:bound}, this gives us an estimate of the length of an optimal tour.
We separately consider the trivial case in which the bounding box, $W_0$ of an optimal solution
contains no vertex of any input region $P_i$ (see the earlier footnote).  
Then, we consider each possible choice of a vertex $c_0$, assumed to lie within $W_0$, 
and consider the $N$-by-$N$ grid ${\cal G}$ centered on $c_0$ (with $N=O(n^2/\epsilon)$);
all computations will now take place with respect to this grid; Lemma~\ref{lem:localization} justifies
this localization step.  For each choice of axis-aligned (grid-conforming) rectangle $W_0$ 
that intersects or contains every input region, we let 
${\cal R}_{W_0}$ denote the regions that are within $W_0$, and let $\bar{\cal R}$
be the remaining ``border'' regions.

A subproblem is defined by a rectangle $W\subseteq W_0$ (whose coordinates
are among those of the grid points ${\cal G}$), together with a 
specification of {\em boundary information}
that gives the information necessary to 
describe how the solution inside $W$ interfaces with 
the solution outside the window $W$.
This information includes the following:
\begin{description}
\item[(a)] For each of the four sides of $W$, we specify a ``bridge''
  segment (on the grid) and at most $2m$ other segments (each with endpoints among ${\cal
    G}$) that cross the side; this is done exactly as in the case of
  the Euclidean TSP on points, as in \cite{m-gsaps-99}.  There are
  $n^{O(m)}$ choices for this information. 
\item[(b)] For each of the four sides of $W$, there is a ``region
  bridge'' segment (corresponding to the $M$-region-span, with endpoints on the grid), and, for
  each of the $2M$ regions of ${\cal R}_{W_0}$ that are not
  intersected by the region bridge segment, we specify (in a single
  bit) whether the region is to be visited (at a grid point of the
  corresponding $\Gamma_i$) within the subproblem or not (if not, it
  is visited outside the window $W$).  
Also, for each of the up to four region bridges, we specify
one of the regions (the ``marked'' region for the bridge) crossed by the bridge and specify for it, 
in a single bit, whether the region is visited inside or outside the subproblem.
There are $n^{O(1)}$ choices for the region bridges (and marked regions) and
  $2^{8M+4}=n^{O(1/\eps)}$ choices for the additional bits.
\item[(c)] For each of the four sides of $W$ there may be regions of
  ${\cal R}\setminus {\cal R}_{W_0}$ that protrude from outside $W_0$
  into the subproblem $W$.  For each such region, we need to specify
  whether or not the subproblem is responsible to visit the region.
  However, there could be far too many ($\Omega(n)$) such regions.  We
  cannot afford to specify each region individually.  The key property
  of these ``protruders from the outside'' is this: They must extend
  all the way from the boundary of $W_0$ across the boundary of $W$.
  
  Consider the left side of $W$.  On this side there are possibly two
  bridging segments (the bridge and the region bridge) specified, as
  well as up to $2m$ specified edges, $e_1,\ldots,e_K$, that cross the
  side and are part of the information specified in (a).  Any region
  of ${\cal R}\setminus {\cal R}_{W_0}$ that intersects one of these
  bridge segments or one of these specified crossing edges is already
  visited by the set $E$ of edges.  There remains a set ${\cal
    R}'\subseteq {\cal R}\setminus {\cal R}_{W_0}$ of other regions
  protruding from outside $W_0$ that intersect the side of $W$ {\em
    between} the bridge segments and specified crossing segments.
  Since we are assuming that regions are disjoint, the set ${\cal R}'$
  forms an ordered set of noncrossing regions extending between the
  boundary of $W$ and the boundary of the root window $W_0$.  Consider
  the subsequence, $P_1^i,\ldots,P_{j_i}^i$ of such regions that
  extend across the subsegment of the wall bounded by $e_i$ and
  $e_{i+1}$. See Figure~\ref{fig:subproblem}.  Because the edge set
  $E$ is connected and lies entirely within $W_0$, we obtain that the
  subset of this sequence that is visited outside our subproblem is
  succinctly representable as a pair of subsequences:

\begin{figure}[htbp]
 \centering{\includegraphics[width=\columnwidth]{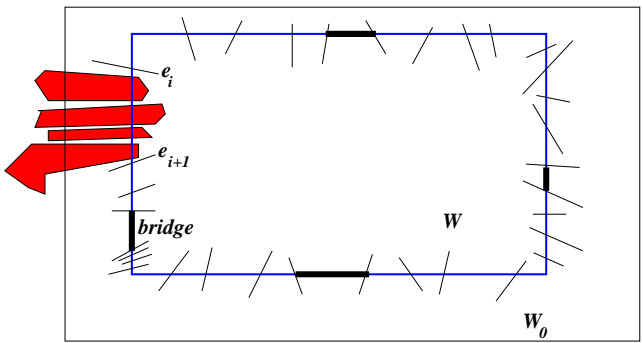}}
\caption{The subproblem defined by window $W$ within the
  bounding rectangle $W_0$.  Some of the edges crossing the boundary
  of $W$ are shown, as are the bridges. (The region bridges are not
  shown, in order not to clutter the diagram.)  The red (shaded) regions shown
  are those regions ${\cal R}'\subseteq {\cal R}\setminus {\cal
    R}_{W_0}$ that protrude from outside $W_0$ in between two
  consecutive segments $e_i$ and $e_{i+1}$ that are part of the
  boundary information for the left wall of $W$.}
 \label{fig:subproblem}
 \end{figure}

\begin{lemma}
  The subset of $\{P_1^i,\ldots,P_{j_i}^i\}$ that is visited by
  portions of $E$ external to $W$ is of the form
  $\{P_1^i,\ldots,P_k^i\}\cup\{P_l^i,\ldots,P_{j_i}^i\}$, for $1\leq
  k\leq l\leq j_i$.
\end{lemma}


Thus, our subproblem can afford to specify, for each pair
$(e_i,e_{i+1})$ along each side of $W$ which subsequence of the
regions protruding from outside are the ``responsibility'' of the
subproblem to visit.

\item[(d)] We specify a required ``connection pattern'' within $W$.
  In particular, we indicate which subsets of the $O(m)$ bridge segments and specified
  edges crossing the boundary of $W$ are required to be connected
  within $W$.  (This is done exactly as is detailed for the
  Euclidean TSP on point sets in \cite{m-gsaps-99}.)
\end{description}

The algorithm produces an optimal $(m,M)$-guillotine network that
satisfies the required constraints, to visit all unbridged regions (as
well as one region associated with each region bridge), and to obey
connectivity constraints.  However, the connectivity constraints do
not explicitly require that every region bridge be connected in with
the rest of the network, and we rely on the region bridge segments to
assure that {\em all} regions are indeed visited by the connected
network.  Thus, at the end of the algorithm, we postprocess our
computed network to ensure global connectivity.  As in
\cite{dm-aatsp-03}, we do this simply as follows, augmenting the
computed network to make each region bridge connected to it.  Working
bottom-up in the hierarchy, we take two sibling subproblems and
consider the region bridge (if any) along the cut between the
subproblems.  We add to the network the boundary of the marked region
associated with the region bridge; by fatness, the perimeter of its
bounding box is at most a constant times greater than the length of
the region bridge.  Further, since the connectivity constraints
required that the marked region be visited (on one side or the other
of the cut) by the network, we know that adding the boundary of the
bounding box of the marked region enforces that the region bridge is
connected to the network.  The total length added in this process is
at most proportional to the lengths of the region bridges; since the
charging scheme ensures that the region bridges need not be more than
$O(\epsilon L^*)$, we know we can afford to add these cycles around
the bounding boxes of marked regions. 

In order to end up with a graph having an Eulerian subgraph spanning
the regions, we use the same trick as done in \cite{m-gsaps-99}: we
``double'' the bridge segments, as well as the region bridge segments, and then
require that the number of connections on each side of a bridge
segment satisfy a parity condition (specified as part of the subproblem).  Exactly as in
\cite{m-gsaps-99}, this allows us to extract a tour from the
planar graph that results from the dynamic programming algorithm
(which gives a shortest possible graph that obeys the specified
conditions).   
The doubled region bridge segments allow the postprocessed network
to preserve the Eulerian property.
  
The result is that in polynomial time ($n^{O(m)}$) one can compute
a shortest possible graph, from a special class of such graphs, and
this graph spans the regions ${\cal R}$, and is Eulerian,
so we can extract a tour.

\remark The running time can be improved to $O(n^{C})$, for a constant
$C$ independent of $1/\eps$, using the method of ``grid-rounded
guillotine subdivisions,'' developed in~\cite{m-gsaps-99,m-gsaps-97}.

\section{Conclusion}
\label{sec:conclusion}

One immediate generalization of our main result is to a special case
of disconnected regions (as in Feremans and Grigoriev~\cite{fg-asggp-05}
and Elbassioni
et al.~\cite{efms-aaegtsp-05}):
We can allow the
regions to be visited to be {\em sets} of points/polygons, each of which
lies within a polygon $P_i$, where the (connected) polygons $P_i$ are fat and disjoint.

Another generalization for which our results give a PTAS is the
$k$-TSPN, in which an integer $k$ is specified and the objective is to
find a shortest tour that visits $k$ regions.

Several open problems remain, including
\begin{description}
\item[(1)] Is there a constant-factor approximation algorithm for
arbitrary connected regions in the plane?  (If the diameters of the
regions are comparable, there are $O(1)$-approximations
known~\cite{dm-aatsp-03,efs-atspin-06}.)  What if the regions are
disconnected?  (giving us a geometric version of a ``one-of-a-set
TSP'')
\item[(2)] What approximation bounds can be obtained in higher dimensions?
A particularly intriguing special case is the generalization of the case of
infinite straight lines: What can be said in 3-space for the TSPN on a set 
of lines or of planes?
\item[(3)] Is there a PTAS for general pairwise-disjoint regions in the plane?
The known APX-hardness proofs rely on regions that may overlap.
\end{description}

\subsection*{Acknowledgments}

I thank Eyal Ackerman, Otfried Cheong, Khaled Elbassioni, and Rene Sitters
for helpful comments and corrections 
on an earlier draft. 

\bibliographystyle{abbrv}

\end{document}